\pdfoutput=1 
\documentclass[12pt]{elsarticle}

\usepackage{graphicx}
\usepackage{wrapfig}

\usepackage{amssymb}
\usepackage{amsmath}

\usepackage{xspace}

\newtheorem{defin}{Definition}
  
\newtheorem{theo}[defin]{Theorem}
  \newenvironment{theorem}{\begin{theo} \sl}{\end{theo}}
\newtheorem{lem}[defin]{Lemma}
  \newenvironment{lemma}{\begin{lem} \sl}{\end{lem}}
\newtheorem{coro}[defin]{Corollary}
  \newenvironment{corollary}{\begin{coro} \sl}{\end{coro}}
\newenvironment{proof}{\emph{Proof.}}{\hfill $\Box$ \medskip\\}

\newcommand{\etal}{\emph{et~al.}\xspace}
\newcommand{\from}[1]{{\emph{\textbf{(#1)}}}}

\journal{Computational Geometry: Theory and Applications}

\begin{document}

\begin{frontmatter}



\title{Making triangulations 4-connected using flips\tnoteref{thanks}}
\tnotetext[thanks]{This research was partially supported by NSERC, GraDR EUROGIGA project No. GIG/11/E023, MTM2009-07242, Gen. Cat. DGR 2009SGR1040 and ESF EUROCORES programme EuroGIGA, CRP ComPoSe: Fonds National de la Recherche Scientique (F.R.S.-FNRS) - EUROGIGA NR 13604.}

\author[carleton]{Prosenjit Bose}
\ead{jit@scs.carleton.ca}

\author[carleton]{Dana Jansens}
\ead{dana@cg.scs.carleton.ca}

\author[carleton]{Andr\'e van Renssen}
\ead{andre@cg.scs.carleton.ca}

\author[maria]{Maria Saumell}
\ead{maria.saumell.m@gmail.com}

\author[carleton]{Sander Verdonschot}
\ead{sander@cg.scs.carleton.ca}

\address[carleton]{School of Computer Science, Carleton University, 5302 Herzberg Laboratories,\\
1125 Colonel By Drive, Ottawa, Ontario, K1S 5B6, Canada.}
\address[maria]{Computer Science Department, Universit\'e Libre de Bruxelles,\\
Boulevard du Triomphe - CP 212, 1050 Brussels, Belgium.}

\begin{abstract}
We show that any combinatorial triangulation on $n$ vertices can be transformed into a 4-connected one using at most $\lfloor(3n - 9)/5\rfloor$ edge flips. We also give an example of an infinite family of triangulations that requires this many flips to be made 4-connected, showing that our bound is tight. In addition, for $n \geq 19$, we improve the upper bound on the number of flips required to transform any 4-connected triangulation into the canonical triangulation (the triangulation with two dominant vertices), matching the known lower bound of $2n - 15$. Our results imply a new upper bound on the diameter of the flip graph of $5.2 n - 33.6$, improving on the previous best known bound of $6n - 30$.
\end{abstract}


\end{frontmatter}


\section{Introduction}
\label{sec:introduction}

\noindent Given a triangulation (a maximal planar simple graph) on a set of $n$ vertices, we define an \emph{edge flip} as removing an edge $(a,b)$ from the graph and replacing it with the edge $(c, d)$, where $c$ and $d$ are the other vertices of the two triangular faces that had $(a, b)$ as an edge. Flipping an edge is allowed if and only if it does not create a duplicate edge. Figure~\ref{fig:flip} shows an example.

\begin{figure}[b]
 \centering
 \includegraphics{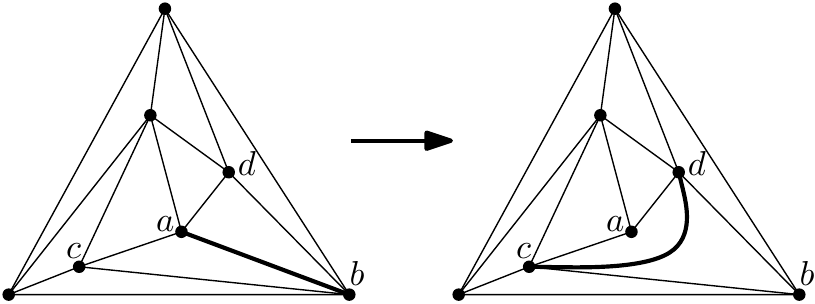}
 \caption{An example triangulation before and after flipping edge $(a, b)$.}
 \label{fig:flip}
\end{figure}

Flips have been studied mostly in two different settings: the \emph{geometric} setting, where we are given a fixed set of points in the plane and edges are straight line segments, and the \emph{combinatorial} setting, where we are only given the clockwise order of edges around each vertex (a combinatorial embedding). In this paper, we concern ourselves with the number of flips required to transform one triangulation into another in the combinatorial setting. We give a brief overview of previous work on this problem. A more detailed overview, including full proofs of the previous bounds, can be found in a recent survey by Bose and Verdonschot~\cite{bose2012history}. For a broader overview on the topic of flips, including applications and related work, we refer the reader to a survey by Bose and Hurtado~\cite{bose2009flips}.

Given a set of $n$ vertices, we define its \emph{flip graph} as the graph with a vertex for each distinct triangulation and an edge between two vertices if their corresponding triangulations differ by a single flip. Two triangulations are considered distinct if they are not isomorphic. In his seminal paper, Wagner~\cite{wagner1936bemerkungenzum} showed that there always exists a sequence of $O(n^2)$ flips that transforms a given triangulation into any other triangulation on the same number of vertices. In terms of the flip graph, Wagner showed that it is connected and has diameter $O(n^2)$. Recently, efforts have been made to find better bounds on the diameter of the flip graph. Komuro~\cite{komuro1997diagonal} was the first to show that the diameter is linear and Mori~\etal~\cite{mori2003diagonal} currently have the strongest upper bound of $6n - 30$.

These results all show how to transform any triangulation into a fixed \emph{canonical} triangulation, which is the  triangulation with two dominant vertices (adjacent to every other vertex). Transformation of one triangulation into another is then straightforward. We transform the first triangulation into the canonical one and transform it into the second triangulation by reversing the sequence of flips that transforms the second triangulation into the canonical one. Mori~\etal's algorithm to transform a triangulation into the canonical one consists of two steps. They first make the given triangulation 4-connected using at most $n - 4$ flips. Since a 4-connected triangulation is always Hamiltonian~\cite{whitney1931theorem}, they then show how to transform any Hamiltonian triangulation into the canonical one by at most $2n - 11$ flips, using a decomposition into two outerplanar graphs that share a Hamiltonian cycle as their outer faces.

The problem of making triangulations 4-connected has also been studied in the setting where many edges may be flipped simultaneously. Bose~\etal~\cite{bose2006simultaneous} showed that any triangulation can be made 4-connected by one simultaneous flip and that $O(\log n)$ simultaneous flips are sufficient and sometimes necessary to transform between two given triangulations.

In Section~\ref{sec:ub}, we show that any triangulation can be made 4-connected using at most $\lfloor(3n - 9)/5\rfloor$ flips. This improves the first step of the construction by Mori~\etal For $n \geq 19$, we also improve the bound on the second step of their algorithm to match the lower bound by Komuro~\cite{komuro1997diagonal}. This results in a new upper bound on the diameter of the flip graph of $5.2 n - 33.6$. We then show in Section~\ref{sec:lb} that, when $n$ is a multiple of 5, there are triangulations that require $(3n - 10)/5 = \lfloor(3n - 9)/5\rfloor$ flips to be made 4-connected, showing that our bound is tight. Section~\ref{sec:lemmas} contains proofs for various technical lemmas that are used in the proof of the upper bound. Section~\ref{sec:conclusions} contains a discussion of our results and some remaining open problems.

\section{Upper bound}
\label{sec:ub}

\noindent In this section we prove an upper bound on the number of flips required to make any given triangulation 4-connected. Specifically, we show that \mbox{$\lfloor(3n - 9)/5\rfloor$} flips always suffice. The proof references several technical lemmas whose proofs can be found in Section~\ref{sec:lemmas}. We also prove that any 4-connected triangulation can be transformed into the canonical form using a worst-case optimal number of $2n - 15$ flips.

We are given a triangulation $T$, along with a combinatorial embedding specifying the clockwise order of edges around each vertex of $T$. In addition, one of the faces of $T$ is marked as the \emph{outer face}. If an edge of the outer face is flipped, one of the two new faces is designated as the new outer face. A \emph{separating triangle} $D$ is a cycle in $T$ of length three whose removal splits $T$ into two (non-empty) connected components. We call the component that contains vertices of the outer face the \emph{exterior} of $D$, and the other component the \emph{interior} of $D$. A vertex in the interior of $D$ is said to be \emph{inside} $D$ and likewise, a vertex in the exterior of $D$ is \emph{outside} $D$. An edge is inside a separating triangle if at least one of its endpoints is inside. A separating triangle $A$ \emph{contains} another separating triangle $B$ if and only if the interior of $B$ is a subgraph of the interior of $A$ with a strictly smaller vertex set. If $A$ contains $B$, $A$ is called the \emph{containing} triangle. A separating triangle that is contained by the largest number of separating triangles in $T$ is called \emph{deepest}. Since containment is transitive, a deepest separating triangle cannot contain any separating triangles, as these would have a higher number of containing triangles. Finally, we call an edge that does not belong to any separating triangle a \emph{free edge}. Free edges have the following useful property.

\begin{lemma}
 \label{lem:freeedge}
 In a triangulation, every vertex $v$ of a separating triangle $D$ is incident to at least one free edge inside $D$.
\end{lemma}
\begin{proof}
 Consider one of the edges of $D$ incident to $v$. Since $D$ is separating, its interior cannot be empty and since $D$ is part of a triangulation, there is a triangular face inside $D$ that uses this edge. Now consider the other edge $e$ of this face that is incident to $v$.

 The remainder of the proof is by induction on the number of separating triangles contained in $D$. For the base case, assume that $D$ does not contain any other separating triangles. Then $e$ must be a free edge and we are done.

 For the induction step, there are two further cases. If $e$ does not belong to a separating triangle, we are again done, so assume that $e$ belongs to a separating triangle $D'$. Since $D'$ is itself a separating triangle contained in $D$ and containment is transitive, the number of separating triangles contained in $D'$ must be strictly smaller than the number contained in $D$. Since $v$ is also a vertex of $D'$, our induction hypothesis tells us that there is a free edge incident to $v$ inside $D'$. Since $D'$ is contained in $D$, this edge is also inside $D$.
\end{proof}

We will remove all separating triangles from $T$ by repeatedly flipping an edge of a separating triangle. This makes $T$ 4-connected, as a triangulation is 4-connected if and only if it has no separating triangles. This technique was also used by Mori~\etal~\cite{mori2003diagonal}, who proved the following lemma.

\begin{lemma}
 \label{lem:flip}
 \from{Mori~\etal~\cite{mori2003diagonal}}
 In a triangulation on $n \geq 6$ vertices, flipping any edge of a separating triangle $D$ will remove that separating triangle. This never introduces a new separating triangle, provided that the selected edge belongs to multiple separating triangles or none of the edges of $D$ belong to multiple separating triangles.
\end{lemma}

\noindent With this in mind, our algorithm works as follows.
\\[0.6\baselineskip]
\noindent\textbf{Algorithm} \textsc{Make 4-connected}
\vspace{-0.6\baselineskip}
\begin{itemize}
 \renewcommand{\labelitemii}{$\circ$}
 \setlength{\itemsep}{-0.6\baselineskip}
 \item Find a deepest separating triangle $D$, preferring ones that do not use an edge of the outer face.
 \vspace{-0.6\baselineskip}
 \begin{itemize}
 \setlength{\itemsep}{-4pt}
 \item If $D$ does not share any edge with other separating triangles, flip an edge of $D$ that is not on the outer face.
 \item If $D$ shares exactly one edge with another separating triangle, flip this edge.
 \item If $D$ shares multiple edges with other separating triangles, flip one of the shared edges that is not shared with a containing triangle.
 \end{itemize}
 \item Repeat until $T$ is 4-connected.
\end{itemize}

\noindent We are now ready prove that this algorithm achieves our new upper bound.

\begin{theorem}
 \label{thm:4-connected}
 A triangulation on $n \geq 6$ vertices can be made 4-connected using at most $\lfloor(3n - 9)/5\rfloor$ flips.
\end{theorem}
\begin{proof}
 We prove this using a charging scheme. We begin by placing a coin on every edge of the triangulation. Then we flip the edges indicated by the algorithm until no separating triangles remain, while paying 5 coins for every flip. The exact charging scheme will be described later. During this process, we maintain two invariants:
 \begin{itemize}
  \item Every edge of a separating triangle has a coin.
  \item Every vertex of a separating triangle has an incident free edge that is inside the triangle and has a coin.
 \end{itemize}

 These invariants have several nice properties. First, an edge can either be a free edge or belong to a separating triangle, but not both. So at any given time, only one invariant applies to an edge. Second, an edge only needs one coin to satisfy the invariants, even if it is on multiple separating triangles or is a free edge for multiple separating triangles. These two properties imply that the invariants hold initially, since by Lemma~\ref{lem:freeedge}, every vertex of a separating triangle has an incident free edge.
 
 We now show that these invariants are sufficient to guarantee that we can pay 5 coins for every flip. Consider the situation after we flip an edge that belongs to a deepest triangle $D$ and satisfies the criteria of Lemma~\ref{lem:flip}, but before we remove any coins. Since flipping the edge has removed $D$ and no new separating triangles are introduced, both invariants still hold. We proceed by identifying four types of edges whose coins we can now remove to pay for this flip without upsetting the invariants.
 
  \smallskip
  \textbf{Type 1 (\includegraphics{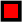}).} The flipped edge $e$. By Lemma~\ref{lem:flip}, $e$ cannot belong to any separating triangle after the flip, so the first invariant still holds if we remove $e$'s coin. Before the flip, $e$ was not a free edge, so the second invariant was satisfied even without $e$'s coin. Since the flip did not introduce any new separating triangles, this is still the case.

  \smallskip
  \textbf{Type 2 (\includegraphics{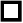}).} A non-flipped edge $e$ of $D$ that is not shared with any other separating triangle. By Lemma~\ref{lem:flip}, the flip removed $D$ and did not introduce any new separating triangles. Therefore $e$ cannot belong to any separating triangle, so the first invariant still holds if we remove $e$'s coin. By the same argument as for the previous type, $e$ is also not required to have a coin to satisfy the second invariant.

  \smallskip


  \textbf{Type 3 (\includegraphics{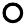}).} A free edge $e$ of a vertex of $D$ that is not shared with any containing separating triangle. Since $e$ did not belong to any separating triangle and the flip did not introduce any new ones, $e$ is not required to have a coin to satisfy the first invariant. Further, since the flip removed $D$ and $D$ was deepest, $e$ is not incident to a vertex of another separating triangle that contains it. Therefore it is no longer required to have a coin to satisfy the second invariant.

 \begin{figure}[ht]
  \centering
  \includegraphics{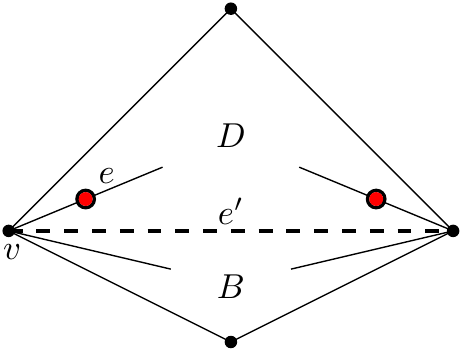}
  \caption{Two type 4 edges.}
  \label{fig:superfluous}
 \end{figure}

  \smallskip
  \textbf{Type 4 (\includegraphics{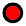}).} A free edge $e$ incident to a vertex $v$ of $D$, where $v$ is an endpoint of an edge $e'$ of $D$ that is shared with a non-containing separating triangle $B$, provided that we flip $e'$ (illustrated in Figure~\ref{fig:superfluous}). Any separating triangle that contains $D$ but not $B$ must share $e'$ (Lemma~\ref{lem:containingshared}) and is therefore removed by the flip.\\
So every separating triangle after the flip that contains $D$ also contains $B$. In particular, this also holds for containing triangles that share $v$. Since the second invariant requires only one free edge with a coin for each vertex of a separating triangle, we can safely charge the one inside $D$, as long as we do not charge the free edge in $B$.
 
 \medskip
 To decide which edges we charge for each flip, we distinguish five cases, based on the number of edges $D$ shares with other separating triangles and whether any of these triangles contain $D$. These cases are illustrated in Figures~\ref{fig:caseA}, \ref{fig:caseC}, and \ref{fig:caseB}.

 \begin{figure}[ht]
  \centering
  \includegraphics{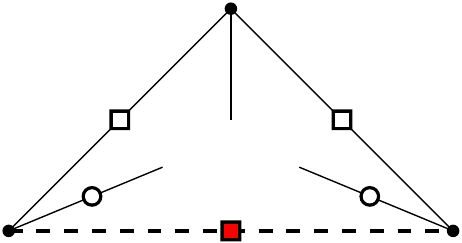}
  \caption{The edges that are charged if the deepest separating triangle does not share any edges with other separating triangles. The flipped edge is dashed and the charged edges are marked with filled boxes (Type~1), empty boxes (Type~2), empty disks (Type~3) or filled disks (Type~4).}
  \label{fig:caseA}
 \end{figure}

 \smallskip
  \textbf{Case 1.} $D$ does not share any edges with other separating triangles (Figure~\ref{fig:caseA}). In this case, we flip any of $D$'s edges. By the first invariant, each edge of $D$ has a coin. These edges all fall into Types 1 and 2, so we use their coins to pay for the flip. Further, $D$ can share at most one vertex with a containing triangle (Lemma~\ref{lem:onecontaining}), so we charge two free edges, each incident to one of the other two vertices (Type~3). 
\newpage
 \begin{figure}[ht]
  \centering
  \includegraphics{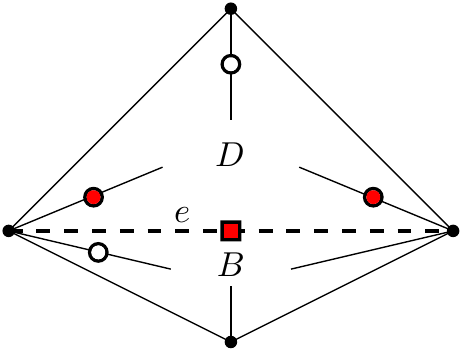}
  \caption{The edges that are charged if the deepest separating triangle only shares edges with non-containing separating triangles.}
  \label{fig:caseC}
 \end{figure}

  \textbf{Case 2.} $D$ does not share any edge with a containing triangle, but shares one or more edges with non-containing separating triangles (Figure~\ref{fig:caseC}). In this case, we flip one of the shared edges $e$. We charge $e$ (Type~1) and two free edges inside $D$ that are incident to the vertices of $e$ (Type~4). This leaves us with two more coins that we need to charge.

  Let $B$ be the non-containing separating triangle that shares $e$ with $D$. We first show that $B$ must have the same depth as $D$. There can be no separating triangles that contain $D$ but not $B$, as any such triangle would have to share $e$ (Lemma~\ref{lem:containingshared}) and $D$ does not share any edge with a containing triangle. Therefore any triangle that contains $D$ must contain $B$ as well. Since $D$ is contained in the maximal number of separating triangles, this holds for $B$ as well. This means that $B$ cannot contain any separating triangles and to satisfy the second invariant we only need to concern ourselves with triangles that contain both $B$ and $D$.

  Now consider the number of vertices of the quadrilateral formed by $B$ and $D$ that can be shared with containing triangles. Since $D$ does not share an edge with a containing triangle, it can share at most one vertex with a containing triangle (Lemma~\ref{lem:onecontaining}). Now suppose that $B$ shares an edge with a containing triangle. Then one of the vertices of this edge is part of $D$ as well. Since the other two vertices of the quadrilateral are both part of $D$, they cannot be shared with containing triangles. On the other hand, if $B$ does not share an edge with a containing triangle, it too can share at most one vertex with containing triangles. Thus, in both cases, at most two vertices of the quadrilateral can be shared with containing triangles, which means that there are at least two vertices that are not shared. For each of these vertices, if it is the vertex of $D$ that is not shared with $B$, we charge the free edge in $D$, otherwise we charge the free edge in $B$ (both Type~3).

 \begin{figure*}[ht]
  \centering
  \includegraphics{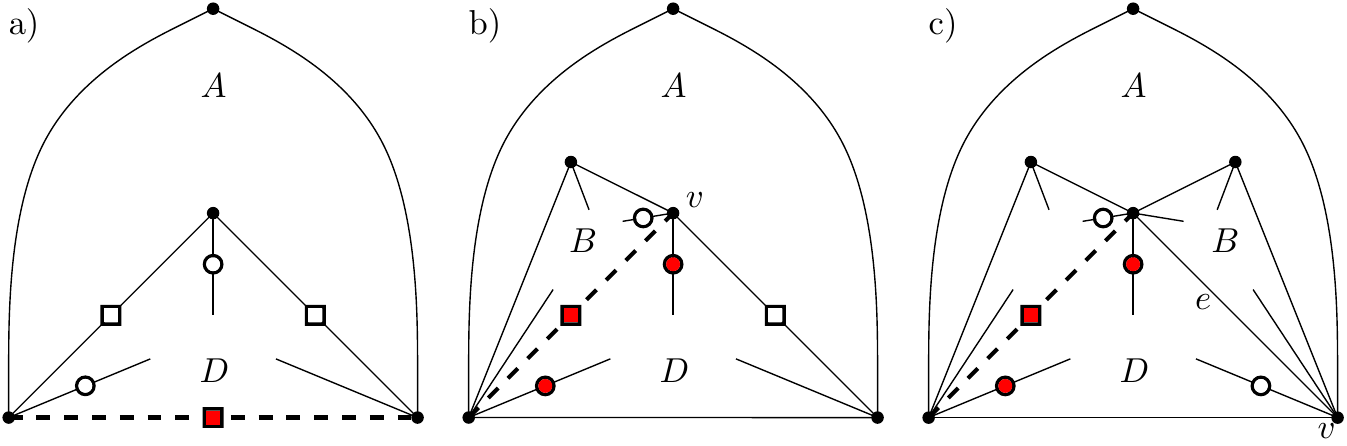}
  \caption{The edges that are charged if the deepest separating triangle shares an edge with a containing triangle.}
  \label{fig:caseB}
  \vspace{1em}
 \end{figure*}

 \smallskip
  \textbf{Case 3.} $D$ shares an edge with a containing triangle $A$ and does not share the other edges with any separating triangle (Figure~\ref{fig:caseB}a). In this case, we flip the shared edge and charge all of $D$'s edges, since one is the flipped edge (Type~1) and the others are not shared (Type~2). The vertex of $D$ that is not shared with $A$ cannot be shared with any containing triangle (Lemma~\ref{lem:unsharedvertex}), so we charge a free edge incident to this vertex (Type~3).

  Further, if $A$ shares an edge with a containing triangle, it either shares the flipped edge, which means that the containing triangle is removed by the flip, or it shares another edge, in which case the vertex that is not an endpoint of this edge cannot be shared with any containing triangle. If $A$ does not share an edge with a containing triangle, it can share at most one vertex with a containing triangle (Lemma~\ref{lem:onecontaining}). In both cases, one of the vertices of the flipped edge is not shared with any containing triangle (Type~3), so we charge a free edge incident to it.

 \smallskip
  \textbf{Case 4.} $D$ shares an edge with a containing triangle $A$ and exactly one other edge with a non-containing separating triangle $B$ (Figure~\ref{fig:caseB}b). In this case, we flip the edge that is shared with $B$. Let $v$ be the vertex of $D$ that is not shared with $A$. We charge the flipped edge (Type~1), the unshared edge of $D$ (Type~2) and two free edges inside $D$ that are incident to the vertices of the flipped edge (Type~4). We charge the last coin from a free edge in $B$ that is incident to $v$. We can charge it, since $v$ cannot be shared with a triangle that contains $D$ (Lemma~\ref{lem:unsharedvertex}) and every separating triangle that contains $B$ but not $D$ must share the flipped edge as well (Lemma~\ref{lem:containingshared}) and is therefore removed by the flip.

  All that is left is to argue that there can be no separating triangle contained in $B$ that requires the coin on this free edge to satisfy the second invariant. Every separating triangle that contains $D$ but not $B$ must share the flipped edge (Lemma~\ref{lem:containingshared}). Since $D$ already shares another edge with a containing triangle and it cannot share two edges with containing triangles (Lemma~\ref{lem:onecontainingtriangle}), all separating triangles that contain $D$ must also contain $B$. Since $D$ is deepest, $B$ must be deepest as well and therefore cannot contain any separating triangles.

 \smallskip
  \textbf{Case 5.} $D$ shares one edge with a containing triangle $A$ and the other two with non-containing separating triangles (Figure~\ref{fig:caseB}c). In this case we also flip the edge shared with one of the non-containing triangles. The charged edges are identical to the previous case, except that there is no unshared edge any more. Instead, we charge the last free edge in $D$.

  Before we argue why we are allowed to charge it, we need to give some names. Let $e$ be the edge of $D$ that is not shared with $A$ and is not flipped. Let $B$ be a non-containing triangle that shares $e$ with $D$ and let $v$ be the vertex that is shared by $A$, $B$ and $D$. Now, any separating triangle that shares $v$ and contains $D$ must contain $B$ as well. If it did not, it would have to share $e$ with $D$, but $D$ already shares an edge with a containing triangle and cannot share more than one (Lemma~\ref{lem:onecontainingtriangle}). Since the second invariant requires only a single free edge with a coin for each vertex of a separating triangle, it is enough that $v$ still has an incident free edge with a coin in $B$.
 
 \medskip
 This shows that we can charge 5 coins for every flip while maintaining the invariants, but we still need to show that after performing these flips we have indeed removed all separating triangles. So suppose that our graph contains separating triangles. Since each separating triangle is contained in a certain number of other separating triangles (which can be zero), there is at least one deepest separating triangle $D$. Since $D$ shares at most one edge with containing separating triangles (Lemma~\ref{lem:onecontainingtriangle}), one of the cases above must apply. This gives us an edge of $D$ to flip and five edges to charge, each of which is guaranteed by the invariants to have a coin. Therefore the process stops only after all separating triangles have been removed.
 
 Finally, since we pay 5 coins per flip and there are $3n - 6$ edges, by initially placing a coin on each edge, we flip at most $\lfloor(3n - 6)/5\rfloor$ edges.
 Now consider the edges of the outer face. We show that these still have a coin at the end of the algorithm. By definition, these edges are not inside any separating triangle and since we only charge free edges inside separating triangles, they can only ever be charged as Type~1 or 2. Thus, if an edge of the outer face gets charged, it was part of the deepest separating triangle $D$ that was removed by the flip. Since an edge of the outer face cannot be shared with a non-containing separating triangle and it cannot be contained by any separating triangle, it can only be charged in Case~1 or 3. In Case~1, we charge only two of the free edges inside $D$, since there could be a containing separating triangle that shares just a vertex. However, this is not possible if $D$ uses an edge of the outer face (Lemma~\ref{lem:outertriangle}), so we can charge this free edge instead of the edge of the outer face. Since we flip one of the edges that is not on the outer face, after the flip, all edges of the outer face still have their coins. In Case~3 we can charge this remaining free edge for the same reasons. However, since in this case we actually flip the edge of the outer face, we are not done yet. The outer face after the flip consists of the flipped edge, one edge of the current outer face and a current interior edge. Charging the extra free edge guarantees that the flipped edge can retain its coin, but we need to ensure that the current interior edge has a coin as well. Let $A$ be the deepest of the separating triangles that contain $D$. Since it, too, uses an edge of the outer face, $A$ can only be contained in triangles that share this edge (Lemma~\ref{lem:outertriangle}). It also cannot contain any separating triangles other than $D$, as these would be deepest as well and we prefer to remove separating triangles that do not use an edge of the outer face. Therefore there can be no other separating triangle that uses the free edge incident to the vertex of $A$ that is not on the outer face and we can move this coin to the new edge of the outer face. Since this is the only case in which an edge of the outer face is flipped, this shows that the edges of the outer face retain their coins during the entire process. Therefore we actually only need $3n - 9$ coins, resulting in a maximum of $\lfloor(3n - 9)/5\rfloor$ flips.
\end{proof}

 \noindent The second step in the algorithm by Mori~\etal~\cite{mori2003diagonal} is to transform the obtained 4-connected triangulation into the canonical form using at most $2n - 11$ flips. We improve this slightly to $2n - 15$ flips, matching the lower bound by Komuro~\cite{komuro1997diagonal}. We first need to prove a few more lemmas.
 
\begin{figure}[ht]
 \centering
 \includegraphics{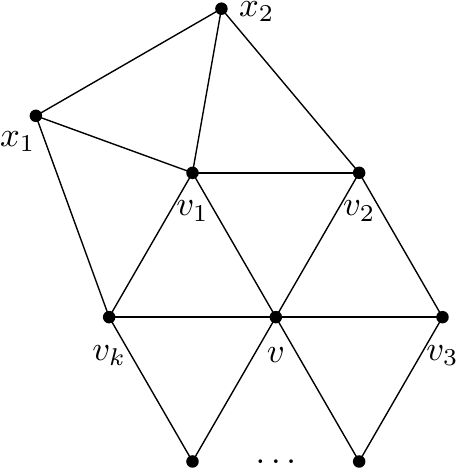}
 \caption{The neighbourhood of a vertex with degree at least 6 in a 4-connected triangulation.}
 \label{fig:degree}
\end{figure}

\begin{lemma}
 \label{lem:degree}
 In a 4-connected triangulation on $n \geq 13$ vertices, every vertex of degree at least 6 either has a neighbour of degree at least 6, or it can be connected to a vertex of degree at least 5 by a single flip.
\end{lemma}
\begin{proof}
 Let $v$ be a vertex of degree at least 6. Komuro~\cite{komuro1997diagonal} showed that either the graph consists of a cycle of length $n - 2$ with $v$ and one other vertex connected to every vertex on the cycle, or $v$ has a neighbour with degree at least 5. In the first case, there is a vertex of high degree that can be connected to $v$ by a single flip, so assume that this is not the case. Let $v_1$ be a neighbour of $v$ with degree at least 5 and let $v_2, \dots, v_k$ be the other neighbours of $v$, in clockwise order from $v_1$. Suppose that none of these neighbours have degree at least 6. Since the graph is 4-connected, this means that each has degree 4 or 5 and $v_1$ has degree exactly 5. Furthermore, no edge can connect two non-consecutive neighbours of $v$, as this would create a separating triangle. Let $x_1$ and $x_2$ be the neighbours of $v_1$ that are not adjacent to $v$, in clockwise order (see Figure~\ref{fig:degree}). We distinguish two cases, based on the degree of $v_2$:
 
 If $v_2$ has degree 4, $x_2$ must be connected to $v_3$. Both $x_1$ and $x_2$ can be connected to $v$ with a single flip, so if either has degree at least 5, we are done. The only way to keep their degree at 4 is to connect both $x_1$ and $v_k$ to $v_3$. But this would give $v_3$ degree at least 6, which is a contradiction. Therefore either $x_1$ or $x_2$ must have degree at least 5.
 
 If $v_2$ has degree 5, let $x_3$ be its new neighbour. Again, if one of $x_1$, $x_2$ or $x_3$ has degree at least 5, we are done. Since $x_2$ already has degree 4, the only way to keep its degree below 5 is to connect $x_1$ and $x_3$ by an edge. But then both $x_1$ and $x_3$ have degree 4 and the only way to keep one at degree 4 is to create an edge to the other. Therefore at least one of $x_1$, $x_2$ or $x_3$ must have degree at least 5.
\end{proof}

\noindent In 1931, Whitney \cite{whitney1931theorem} showed that any 4-connected triangulation has a Hamiltonian cycle. The main ingredient of his proof is the following lemma:

\begin{lemma}
 \label{lem:cyclepath}
 \from{Whitney \cite{whitney1931theorem}}
 Consider a cycle $C$ in a 4-connected triangulation, along with two distinct vertices $a$ and $b$ on $C$. These vertices split $C$ into two paths $C_1$ and $C_2$ with $a$ and $b$ as endpoints. Consider all edges on one side of the cycle, say the inside. If no vertex on $C_1$ (resp. $C_2$) is connected to another vertex on $C_1$ (resp. $C_2$) by an edge inside $C$, we can find a path from $a$ to $b$ that passes through each vertex on and inside $C$ exactly once and uses only edges of $C$ and inside $C$.
\end{lemma}

\noindent We use this to prove the following lemma:

\begin{figure}[ht]
 \centering
 \includegraphics{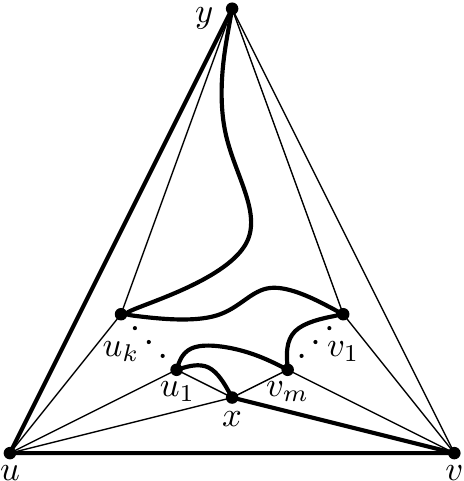}
 \caption{A possible Hamiltonian cycle that uses $(u, v)$ and has all non-cycle edges incident to $u$ on one side of the cycle and all non-cycle edges incident to $v$ on the other.}
 \label{fig:hamiltoniancycle}
\end{figure}

\begin{lemma}
 \label{lem:goodcycle}
 For every edge $(u, v)$ in a 4-connected triangulation, there is a Hamiltonian cycle that uses $(u, v)$ such that all non-cycle edges incident to $u$ are on one side of the cycle and all non-cycle edges incident to $v$ are on the other side.
\end{lemma}
\begin{proof}
 Let $x$ and $y$ be the other vertices of the faces that have $(u, v)$ as an edge. Let $v, x, u_1, \dots , u_k, y$ be the neighbours of $u$ in counter-clockwise order and let $y, v_1, \dots , v_m, x, u$ be the neighbours of $v$ (see Figure~\ref{fig:hamiltoniancycle}). Note that all the $u_i$ and $v_i$ are distinct vertices, as a vertex other than $x$ or $y$ that is adjacent to both $u$ and $v$ would form a separating triangle. This means that $x, u_1, \dots , u_k, y, v_1, \dots , v_m, x$ forms a cycle. Moreover, no two non-consecutive neighbours of $u$ can be connected by an edge, since this would create a separating triangle as well. Since this holds for the neighbours of $v$ as well, $x$ and $y$ split the cycle into two parts that satisfy the conditions of Lemma~\ref{lem:cyclepath}. If we call the side of the cycle that does not contain $(u, v)$ the inside, this means that we can find a path from $x$ to $y$ that passes through each vertex on and inside the cycle exactly once and uses only edges of and inside the cycle. This path can be completed to a Hamiltonian cycle that satisfies the conditions by adding the edges $(y, u)$, $(u, v)$ and $(v, x)$.
\end{proof}

\begin{theorem}
\label{thm:canonical}
Any 4-connected triangulation $T$ on $n \geq 13$ vertices can be transformed into the canonical triangulation using at most $2n - \Delta(T) - 8$ flips, where $\Delta(T)$ is the maximum degree among vertices of $T$.
\end{theorem}
\begin{proof}
We use the same approach as Mori~\etal~\cite{mori2003diagonal}, but instead of taking an arbitrary Hamiltonian cycle, we use the preceding lemmas to carefully construct a good cycle.

Let $x$ be a vertex of maximal degree in $T$ and suppose for now that $x$ has a neighbour $y$ with degree at least 6. We use the cycle given by Lemma~\ref{lem:goodcycle} to decompose $T$ into two outerplanar graphs $T_1$ and $T_2$, each sharing the cycle and having all edges on the inside and outside, respectively. Note that $x$ is an ear in one of these, say $T_2$, while $y$ is an ear in the other. Mori~\etal showed that we can make any vertex $v$ of an outerplanar graph dominant using at most $n - d_v - 1$ flips, where $d_v$ is the degree of $v$. Therefore we can make $x$ dominant in $T_1$ using at most $n - \Delta(T) - 1$ flips. These flips are allowed because $x$ does not have any incident edges in $T_2$. Then we can make $y$ dominant in $T_2$ using at most $n - d_y - 1 \leq n - 7$ flips. Thus we can transform $T$ into the canonical triangulation using at most $2n - \Delta(T) - 8$ flips.

Since any triangulation on $n \geq 13$ vertices has a vertex of degree at least 6, if $x$ does not have a neighbour with degree at least 6, Lemma~\ref{lem:degree} tells us that there is a vertex $v$ with degree at least 5 that can be connected to $x$ by a single flip. We perform this flip and use $v$ in the place of $y$. Since $x$ now has degree $\Delta(T) + 1$, we can make it dominant using at most $n - \Delta(T) - 2$ flips. Similarly, $v$ has degree at least 6 after the flip, so we can make it dominant using at most $n - 7$ flips. Including the initial flip, we again obtain the canonical triangulation using at most $2n - \Delta(T) - 8$ flips.
\end{proof}

\noindent Combining this result with Theorem~\ref{thm:4-connected} gives the following bound on the maximum flip distance between two triangulations.

\begin{corollary}
 Any two triangulations $T_1$ and $T_2$ can be transformed into each other using at most $5.2 n - 19.6 - \Delta(T_1) - \Delta(T_2)$ flips, where $\Delta(T)$ is the maximum degree among vertices of $T$.
\end{corollary}

\noindent Theorem~\ref{thm:canonical} matches the worst-case lower bound of $2n - 15$ flips if the maximum degree is at least 7, but we need a slightly stronger result if the maximum degree is 6.

\begin{lemma}
 \label{lem:6-6flip}
 In a 4-connected triangulation on $n \geq 19$ vertices with maximum degree 6, there is always a pair of vertices of degree 6 that can be connected by a flip.
\end{lemma}
\begin{proof}
 Suppose that such a pair does not exist and consider the neighbourhood of a vertex $v$ of degree 6. Each edge incident to $v$ can be flipped, otherwise there would be an edge connecting two non-consecutive neighbours of $v$, forming a separating triangle. Thus there are 6 pairs of vertices that can be connected by a flip and one vertex of each pair needs to have degree at most 5. To realize this, $v$ needs to have at least 4 neighbours of degree at most 5. Similarly, a vertex of degree 5 needs at least 3 such neighbours and a vertex of degree 4 needs at least 2. Therefore each vertex of degree at most 5 can have at most 2 neighbours of degree 6.

 Let $n_d$ be the number of vertices of degree $d$ and let $k$ be the number of edges between vertices of degree 6 and vertices of degree at most 5. Every vertex of degree 6 needs at least 4 neighbours of degree at most 5, so $k \geq 4 n_6$. But every vertex of degree at most 5 can have at most 2 neighbours of degree 6, so $k \leq 2 (n_4 + n_5)$. Combining these inequalities, we get that $n_6 \leq (n_4 + n_5)/2$. Since a triangulation with maximum degree 6 can have at most 12 vertices of degree less than 6, it follows that $n = n_4 + n_5 + n_6 \leq 18$. Thus for $n \geq 19$, there is always a pair of vertices of degree 6 that can be connected by a flip.
\end{proof}

\begin{theorem}
\label{thm:canonicalmaxdeg6}
Any 4-connected triangulation on $n \geq 19$ vertices with maximum degree 6 can be transformed into the canonical triangulation using at most $2n - 15$ flips.
\end{theorem}
\begin{proof}
By Lemma~\ref{lem:6-6flip}, there is always a pair of vertices $x$ and $y$ of degree 6 that can be connected by a flip. We first perform the flip that connects \mbox{$x$ and $y$}, giving both vertices degree 7. We then proceed similarly to the proof of Theorem~\ref{thm:canonical}. We make $x$ dominant in one of the outerplanar graphs using $n - 8$ flips and we make $y$ dominant in the other, also using $n - 8$ flips. Counting the initial flip, we obtain the canonical triangulation using at most $2n - 15$ flips.
\end{proof}

\noindent By combining this with Theorem~\ref{thm:canonical}, we get the following bound.

\begin{corollary}
 \label{cor:4-contocanon}
 Any 4-connected triangulation $T$ on $n \geq 19$ vertices can be transformed into the canonical triangulation using at most $\min\{2n - 15, 2n - \Delta(T) - 8\}$ flips, where $\Delta(T)$ is the maximum degree among vertices of $T$.
\end{corollary}

\noindent And finally, using our bound from Theorem~\ref{thm:4-connected} on the number of flips it takes to make triangulations 4-connected, we obtain an improved bound on the diameter of the flip graph.

\begin{corollary}
 The diameter of the flip graph of all triangulations on $n \geq 19$ vertices is at most $5.2 n - 33.6$.
\end{corollary}

\section{Lower bound}
\label{sec:lb}

\noindent In this section we present a lower bound on the number of flips that is required to remove all separating triangles from a triangulation. Specifically, we present a triangulation that has $(3n - 10)/5$ edge-disjoint separating triangles, thereby showing that there are triangulations that require this many flips to make them 4-connected.

The triangulation that gives rise to the lower bound is constructed recursively and resembles the Sierpi\'nski triangle~\cite{sierpinski1915suc}. The construction starts with an empty triangle. The recursive step consists of adding an inverted triangle in the interior and connecting each vertex of the new triangle to the two vertices of the opposing edge of the original triangle. This is recursively applied to the three new triangles that share an edge with the inserted triangle, but not to the inserted triangle itself. After $k$ iterations, instead of applying the recursive step again, we add a single vertex in the interior of each triangle we are recursing on and connect this vertex to each vertex of the triangle. We also add a single vertex in the exterior face so that the original triangle becomes separating. The resulting triangulation is called $\mathcal{T}_k$. Figure~\ref{fig:sierpinsky} illustrates this process for $\mathcal{T}_1$ and $\mathcal{T}_2$.

\begin{figure}[ht]
 \centering
 \includegraphics{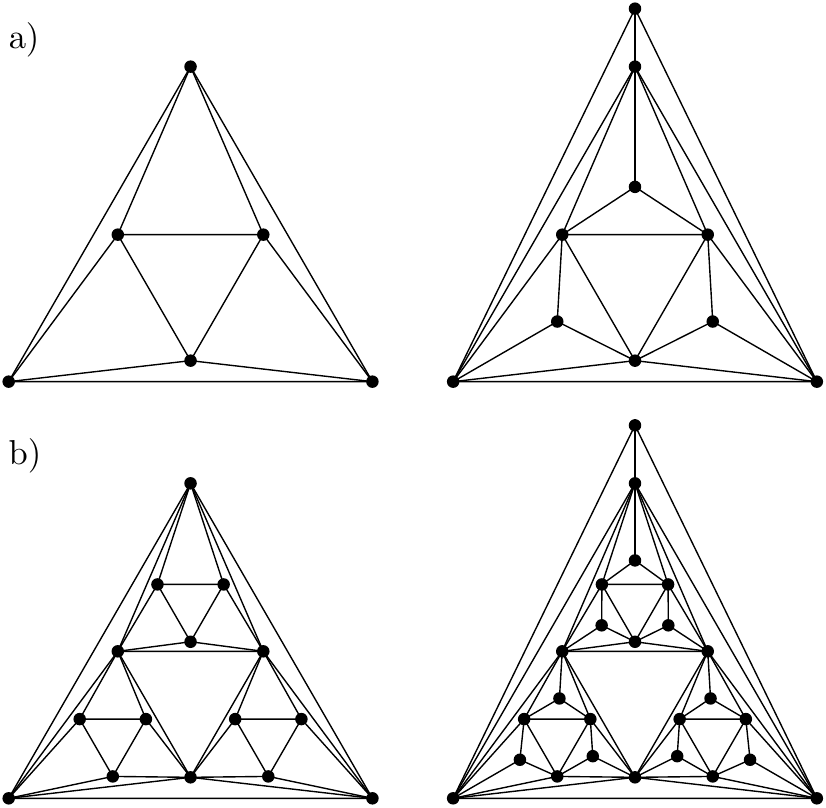}
 \caption{Triangulations $\mathcal{T}_1$ (a) and $\mathcal{T}_2$ (b), before and after the final step of the construction.}
 \label{fig:sierpinsky}
\end{figure}

\begin{theorem}
 There are triangulations that require $(3n - 10)/5$ flips to make them 4-connected, where $n$ is a multiple of 5.
\end{theorem}
\begin{proof}
 In the construction scheme presented above, each of the triangles we recurse on becomes a separating triangle that does not share any edges with the original triangle or the other triangles that we recurse on. Thus all these separating triangles are edge-disjoint. But how many of these triangles do we get? Let $L_i$ be the number of triangles that we recurse on after $i$ iterations of the construction, so $L_0 = 1$, $L_1 = 3$, etc. Now let $V_i$ be the number of vertices of $\mathcal{T}_i$. We can see that $V_1 = 10$ and if we transform $\mathcal{T}_1$ into $\mathcal{T}_2$, we have to remove each of the interior vertices added in the final step and replace them with a configuration of 6 vertices. So to get $\mathcal{T}_2$, we add 5 vertices in each of the $L_1$ triangles. This is true in general, giving
 \begin{equation}
  V_i~~=~~V_{i-1} + 5 L_{i-1}~~=~~10 + 5 \sum_{j=2}^{i} L_{j-1} \label{eq:v}
 \end{equation}
\\ \ \\
 Let $S_i$ be the number of separating triangles of $\mathcal{T}_i$. We can see that $S_1 = 4$ and each recursive refinement of a separating triangle leaves it intact, while adding 3 new ones. Therefore
 \begin{equation}
  S_i~~=~~S_{i-1} + 3 L_{i-1}~~=~~4 + 3 \sum_{j=2}^{i} L_{j-1} \label{eq:s}
 \end{equation}
 From Equation~\eqref{eq:v}, we get that 
 \[ \sum_{j=2}^{i} L_{j-1}~~=~~\frac{V_i - 10}{5} \]
 Substituting this into Equation~\eqref{eq:s} gives 
 \[ S_i~~=~~4 + 3 \cdot \frac{V_i - 10}{5}~~=~~\frac{3 V_i - 10}{5} \]
 Since each flip removes only the separating triangle that the edge belongs to, we need $(3n - 10)/5$ flips to make this triangulation 4-connected. Constructions for multiples of 5 between $V_i$ and $V_{i+1}$ can be obtained by recursing on a subset of the triangles in the final recursion step.
\end{proof}

\noindent Note that this triangulation is not useful for a lower bound on the diameter of the flip graph in general, since it is already Hamiltonian. Thus, even though it is not 4-connected, we know that it can be transformed into the canonical triangulation by at most $2n - 11$ flips from the proof by Mori~\etal~\cite{mori2003diagonal}.

\section{Lemmas and proofs}
\label{sec:lemmas}

\noindent This section contains proofs for the technical lemmas used in the proof of Theorem~\ref{thm:4-connected}.

\begin{lemma}
 \label{lem:interior}
 If a separating triangle $A$ contains a separating triangle $B$, then there is a vertex of $B$ inside $A$ and no vertex of $B$ can lie outside $A$.
\end{lemma}
\begin{proof}
 Let $z$ be a vertex in the interior of $B$ and let $y$ be a vertex of $A$ that is not shared with $B$. Since the interior of $B$ is a subgraph of the interior of $A$ and $y$ is not inside $A$, $y$ must be outside $B$. Since every triangulation is 3-connected, there is a path from $z$ to $y$ that stays inside $A$. This path connects the interior of $B$ to the exterior, so there must be a vertex of $B$ on the path and hence inside $A$.
 
 Now suppose that there is another vertex of $B$ outside $A$. Since all vertices of a triangle are connected by an edge, there is an edge between this vertex and the vertex of $B$ inside $A$. This contradicts the fact that $A$ is a separating triangle, so no such vertex can exist.
\end{proof}
\vspace{-0.5em}
\begin{lemma}
 \label{lem:interiorvertex}
 If a vertex $x$ of a separating triangle $B$ is inside a separating triangle $A$, then $A$ contains $B$.
\end{lemma}
\begin{proof}
 Let $y$ be a vertex of $A$ that is not shared with $B$. There is a path from $y$ to the outer face that stays in the exterior of $A$. There can be no vertex of $B$ on this path, since this would create an edge between the interior and exterior of $A$. Therefore $y$ is outside $B$.
 
 Now suppose that $A$ does not contain $B$. Then there is a vertex $z$ inside $B$ that is not inside $A$. There must be a path from $z$ to $x$ that stays inside $B$. Since $x$ is inside $A$, there must be a vertex of $A$ on this path. But since $y$ is outside $B$, this would create an edge between the interior and exterior of $B$. Therefore $A$ must contain $B$.
\end{proof}
\vspace{-0.5em}
\begin{lemma}
 \label{lem:onecontainingtriangle}
 A separating triangle can share at most one edge with containing triangles.
\end{lemma}
\begin{proof}
 Suppose we have a separating triangle $D$ that shares two of its edges with separating triangles that contain it. First of all, these triangles cannot be the same, since then they would be forced to share the third edge as well, which means that they are $D$. Since a triangle does not contain itself, this is a contradiction. So call one of these triangles $A$ and call one of the triangles that shares the other edge $B$. Let $x$, $y$ and $z$ be the vertices of $D$, such that $x$ is shared with $A$ and $B$, $y$ is shared only with $A$ and $z$ is shared only with $B$.

 By Lemma~\ref{lem:interior}, $z$ must be inside $A$, while $y$ must be inside $B$, since in both cases the other two vertices of $D$ are shared and therefore not in the interior. But then by Lemma~\ref{lem:interiorvertex}, $A$ contains $B$ and $B$ contains $A$. This is a contradiction, since by transitivity it would imply that the interior of $A$ is a subgraph of itself with a strictly smaller vertex set.
\end{proof}
\vspace{-0.5em}
\begin{lemma}
 \label{lem:onecontaining}
 A separating triangle $D$ that shares no edge with containing triangles can share at most one vertex with containing triangles.
\end{lemma}
\begin{proof}
 Suppose that $D$ shares two of its vertices with containing triangles. First, both vertices cannot be shared with the same containing triangle, since then the edge between these two vertices would also be shared. Now let $A$ be one of the containing triangles and let $B$ be one of the containing triangles sharing the other vertex. By Lemma~\ref{lem:interior}, there must be a vertex of $D$ inside $A$. So then both vertices of $D$ that are not shared with $A$ must be inside $A$, otherwise there would be an edge between the interior and the exterior of $A$. In particular, the vertex shared by $B$ and $D$ lies inside $A$, which by Lemma~\ref{lem:interiorvertex} means that $A$ contains $B$. But the reverse is also true, so $B$ contains $A$ as well, which is a contradiction.
\end{proof}
\vspace{-0.5em}
\begin{lemma}
 \label{lem:unsharedvertex}
 A separating triangle that shares an edge with a containing triangle cannot share the unshared vertex with another containing triangle.
\end{lemma}
\begin{proof}
 Suppose we have a separating triangle $D = (x, y, z)$ that shares an edge $(x, y)$ with a containing triangle $A$ and the other vertex $z$ with another containing triangle $B$. By Lemma~\ref{lem:interior}, at least one of $x$ and $y$ has to be inside $B$. Since these are vertices of $A$, by Lemma~\ref{lem:interiorvertex}, $B$ contains $A$. Similarly, $z$ has to be inside $A$ and since it is a vertex of $B$, $A$ contains $B$. This is a contradiction.
\end{proof}
\vspace{-0.5em}
\begin{lemma}
 \label{lem:containingshared}
 Given two separating triangles $A$ and $B$ that share an edge $e$, any separating triangle that contains $A$ but not $B$ must use $e$.
\end{lemma}
\begin{proof}
 Suppose that we have a separating triangle $D$ that contains $A$, but not $B$ and that does not use one of the vertices $v$ of $e$. By Lemma~\ref{lem:interior}, $v$ must be inside $D$. But then $D$ would also contain $B$ by Lemma~\ref{lem:interiorvertex}, as $v$ is a vertex of $B$ as well. Therefore $D$ must share both vertices of $e$ and hence $e$ itself.
\end{proof}
\vspace{-0.5em}
\begin{lemma}
 \label{lem:outertriangle}
 A separating triangle $D$ that uses an edge $e$ of the outer face cannot be contained in a separating triangle that does not share $e$.
\end{lemma}
\begin{proof}
 Suppose $D$ is contained in a separating triangle $A$. If $A$ does not share $e$, by Lemma~\ref{lem:interior}, at least one of the vertices of $e$ must be inside $A$. But since $e$ is part of the outer face, this is a contradiction.
\end{proof}

\section{Conclusions and future work}
\label{sec:conclusions}

\noindent We showed that any triangulation on $n$ vertices can be made 4-connected using at most $\lfloor (3n - 9)/5 \rfloor$ flips, while there are triangulations that require $(3n - 10)/5 = \lfloor (3n - 9)/5 \rfloor$ flips when $n$ is a multiple of 5. This shows that our bound is tight for an infinite family of values for $n$, although a slight improvement to $\lfloor (3n - 10)/5 \rfloor$ is still possible. We believe that this is the true bound. We also showed that any 4-connected triangulation on $n \geq 19$ vertices can be transformed into the canonical form using at most $2n - 15$ flips. This matches the lower bound by Komuro~\cite{komuro1997diagonal} in the worst case where the graph has maximum degree 6 and results in a new upper bound of $5.2 n - 33.6$ on the diameter of the flip graph. It also means that both steps of the algorithm, when considered individually, are now tight in the worst case. Therefore, any further improvement must either merge the two steps in some fashion or employ a different technique.

One potential direction for improvement is to see if fewer flips are required to make triangulations Hamiltonian than to make them 4-connected. While 4-connectivity is a sufficient condition for Hamiltonicity, it is not required. Aichholzer~\etal~\cite{oswin08} gave a lower bound of $(n - 8)/3$ on the number of flips required to make a triangulation Hamiltonian.

Further, all of the current algorithms use the same, single, canonical form. Surprisingly, the lower bound of $2n - 15$ flips to the canonical form is also the best known lower bound on the diameter of the flip graph. This suggests that at least one of the two bounds still has significant room for improvement. So is there another canonical form that gives a better upper bound? Or can we get a better bound by using multiple canonical forms and picking the closest?

Another interesting problem is to minimize the number of flips to make a triangulation 4-connected. We showed that our technique is worst-case optimal, but there are cases where fewer flips would suffice. There is a natural formulation of the problem as an instance of 3-hitting set, where the subsets correspond to the edges of separating triangles and we need to pick a minimal set of edges such that we include at least one edge from every separating triangle. This gives a simple 3-approximation algorithm that picks an arbitrary separating triangle and flips all shared edges or an arbitrary edge if there are no shared edges. However, it is not clear whether the problem is even NP-hard, as not all instances of 3-hitting set can be encoded as separating triangles in a triangulation. Therefore it might be possible to compute the optimal sequence in polynomial time.





\bibliographystyle{unsrt}
\bibliography{papers}







\end{document}